\documentclass[11pt]{article}

\usepackage{amsmath}
\usepackage{amsfonts}
\usepackage{amsthm}
\usepackage{xcolor}
\usepackage{mathtools}
\usepackage{enumerate}
\usepackage{listings}
\usepackage[affil-sl]{authblk}

\definecolor{prompt}{rgb}{0.9,0.0,0.0}
\input listings-python.prf
\numberwithin{equation}{section}

\theoremstyle{plain} \newtheorem{thm}{Theorem}[section]
\theoremstyle{plain} \newtheorem{define}{Definition}[section]
\theoremstyle{plain} \newtheorem{example}{Example}[section]
\theoremstyle{plain} \newtheorem{remark}{Remark}[section]

\newcommand{\subject}{evolutionary biology}
\newcommand{\Additive}{Additive}
\newcommand{\additive}{additive}
\newcommand{\linmodAfterInterpret}{previous section}
\newcommand{\nn}{\nonumber} % For multline label suppression

\newcommand{\opname}{Walsh}
\newcommand{\opnamecoeff}{\opname{} coefficients}
\newcommand{\opnamecoef}{\opname{} coefficient}
\newcommand{\zbopname}{zero-grounded}
\newcommand{\wt}{\opname{} transform}
\newcommand{\zbt}{\zbopname{} transform}
\newcommand{\wb}{\opname{} basis}
\newcommand{\zbb}{\zbopname{} basis}
\newcommand{\wmat}{W}
\newcommand{\zbmat}{Z}

\newcommand{\Is}[3]{{#1}_{#2}\ldots {#1}_{#3}}
\newcommand{\Iss}[1]{\Is{#1}1n}

\newcommand{\vIss}[1]{\left[\Iss{#1}\right]}
\newcommand{\vx}[1]{\left[{#1}\right]}
\newcommand{\sBasis}[1]{\{\vIss{#1}\}_{\Iss{#1} \in \Cnbits}}
\newcommand{\sumIs}[2]{\sum_{\Is{#1}{#2}n}}
\newcommand{\sumIss}[1]{\sum_{\Iss{#1}}}
\newcommand{\sumIsk}[2]{\sum_{\substack{\Is{#1}{1}{{#2}-1},\\\Is{#1}{{#2}+1}{n}}}}
\newcommand{\innerp}[2]{\left< {#1},{#2} \right>}
\newcommand{\bitDotProd}[2]{{#1}_1{#2}_1 + \cdots + {#1}_n{#2}_n}
\newcommand{\negonepwr}[2]{(-1)^{{#1}_1{#2}_1 + \cdots + {#1}_n{#2}_n}}

\newcommand{\negonepwrss}[4]{(-1)^{{#1}_{#3}{#2}_{#3} + \cdots +
{#1}_{#4}{#2}_{#4}}}
\newcommand{\negonepwradd}[3]{(-1)^{{#1}_1({#2}_1 + {#3}_1) + \cdots +
    {#1}_n({#2}_n + {#3}_n)}}
\newcommand{\negonepwraddl}[3]{(-1)^{({#1}_1 + {#2}_1){#3}_1 + \cdots +
    ({#1}_n + {#2}_n){#3}_n}}

\newcommand{\negonepwraddsl}[4]{(-1)^{({#1}_{#4} +
{#2}_{#4}){#3}_{#4} + \cdots +
    ({#1}_n + {#2}_n){#3}_n}}
\newcommand{\negonepwraddss}[5]{(-1)^{{#1}_{#4}({#2}_{#4} +
{#3}_{#4}) + \cdots +
    {#1}_{#5}({#2}_{#5} + {#3}_{#5})}}
\newcommand{\Cnbits}{\mathbb{C}^{2^n}}
\newcommand{\Cnbitsc}[1]{\mathbb{C}^{2^{#1}}}
\newcommand{\otheri}[1]{\tilde{#1}}
\newcommand{\effects}{c_0, c_1, c_2, \ldots, c_n}
\newcommand{\effected}[1]{c_0 + c_1{#1}_1 + c_2{#1}_2 + \cdots + c_n{#1}_n}
\newcommand{\po}{\preceq}
\newcommand{\potriple}[4]{\Is{#1}{1}{#4}\po\Is{#2}{1}{#4}\po\Is{#3}{1}{#4}}
\newcommand{\andop}{{\rm \&}}
\newcommand{\andopSpaced}{{\def\pad{\hskip5pt} \pad\andop\pad}}
\newcommand{\fixed}[1]{\relax^\dag{#1}}
\newcommand{\orderedSumn}[3]{\sum_{\substack{\Iss{#1}\po\\ \Iss{#2}\po\\
\Iss{#3}}}}

\newcommand{\orderedSumnFixed}[3]{\sum_{\substack{\fixed{\Iss{#1}} \po\\\Iss{#2}
\po\\\fixed{\Iss{#3}}}}}
\newcommand{\orderedSumxFixed}[4]{\sum_{\substack{\fixed{\Is{#1}{1}{#4}}
\po\\\Is{#2}{1}{#4} \po\\\fixed{\Is{#3}{1}{#4}}}}}
\newcommand{\Bigskip}{\bigskip}

\newcommand{\wted}[1]{#1^\textbf{w}}
\newcommand{\zbed}[1]{#1^\textbf{z}}

\def\wBasisElt#1{\ifcase#1
        \vx{00} + \vx{10} + \vx{01} + \vx{11}\or
        \vx{00} - \vx{10} + \vx{01} - \vx{11}\or
        \vx{00} + \vx{10} - \vx{01} - \vx{11}\or
        \vx{00} - \vx{10} - \vx{01} + \vx{11}\else
        \text{garbage}\fi}

\def\lastyn#1{\if y#1last\else\fi}

\newcommand{\pandas}{\texttt{pandas}}
\newcommand{\statsmodels}{\texttt{statsmodels}}
\newcommand{\pandaGit}{\texttt{66e3805b8cabe977f40c05259cc3fcf7ead5687d}}
\newcommand{\statsmodelsVersion}{\texttt{0.13.5}}

\newcommand{\codewin}{Listing}

\newcommand{\ack}{\section*{Acknowledgements}}

\newcommand{\tableref}{Table}

\newcommand{\varstyle}[1]{{\tt{#1}}}

\newcommand{\ls}{least-squares}
\newcommand{\Ls}{Least-squares}
\newcommand{\lsr}{\ls{} regression}
\newcommand{\Lsr}{\Ls{} regression}

\author{Devin Greene}
\affil{American University, 4400 Massachusetts Avenue, NW, Washington, DC, United States, 20016}
\title{A Primer for the \wt{}}
\date{}

\begin{document}

\maketitle

\begin{abstract}
A mathematical development of the \wt{}, \wb{}, and \opnamecoeff{} is
given.  The author was prompted to write this by a wish to give a
unified treatment of epistatic coordinates as they are used in
\subject{}.  At the end of the article, opinions are expressed
regarding the usefulness of these concepts for the practical researcher.
\end{abstract}

% \tableofcontents

\section{Introduction}

Since their introduction
\cite{WEINREICH2013700,f622c051-cd98-3197-b98c-f4515ba3978d},
higher-order epistatic coordinates have become a much discussed and used
tool in \subject{}.  The primary form of these coordinates is derived from
a linear unitary transformation called the \wt{}.  The \wt{} has an
associated basis which can be seen as providing a complete description
of the interactions among a set of binary variables.

In this article, we present a mathematical development of the \wt{} and
its associated concepts.  We also introduce the notions of the \zbt{}
and \zbb{}, which can also be used to describe interactions. The
presentation here strives to be as self-contained as possible. However,
the reader is assumed to have had at least the equivalent of an
undergraduate course in linear algebra, as well as knowing what a
partial order is.

There are several strains of ideas in the literature for dealing with
interactions.  This author was prompted to write this work by a desire
to treat these strains as a unified whole.  In Section
\ref{lastSection}, however, we suggest via example that much of the
information one can derive from \opnamecoeff{} can also be derived from
a standard regression analysis with binary input variables, and vice
versa.  We then consider the implications of this for the practical
investigator.

\section{Preliminaries}

Consider the Hilbert space

$$(\Cnbits,\innerp{\cdot}{\cdot})$$
where $\innerp{\cdot}{\cdot}$ is the standard inner product.  Since the
dimension $2^n$ is a power of two, we can index the components of
$\Cnbits$ with bit strings $\Iss{i}$ of length $n$, and do so for the
remainder of this article.  The $\Iss{i}$\,th component of a vector $v
\in \Cnbits$ is written with a subscript: $v_{\Iss{i}}$.  We use the
notation $\vIss{i}$ to denote the vector in $\Cnbits$ which is 1 at
index $\Iss{i}$ and 0 at all other indices. If we use the Kronecker
delta function,

\begin{equation*}
\delta_{x, y} =
\begin{cases}
0& \text{if $x \ne y$},\\
1& \text{if $x = y$},
\end{cases}
\end{equation*}
then we can express the components of $\vIss{i}$ as follows.

$$\vIss{i}_{\Iss{j}} = \delta_{\Iss{i},\Iss{j}}$$

Using these conventions, the inner product of two vectors $v,w \in
\Cnbits$ is given by

$$\innerp{v}{w} = \sumIss{i} v_{\Iss{i}} \overline{w_{\Iss{i}}}$$
where the bar represents complex conjugation. Note that for all bit
strings $\Iss{i}$ and $\Iss{j}$,

$$\innerp{\vIss{i}}{\vIss{j}} = \delta_{\Iss{i},\Iss{j}}$$
Thus the set $\sBasis{i}$ is an orthonormal basis of $\Cnbits$, which is
usually called the \textit{standard basis}. More generally, for all $v
\in \Cnbits$ and $\vIss{i} \in \Cnbits$,

\begin{equation*}
v_{\Iss{i}} = \innerp{v}{\vIss{i}}
\end{equation*}

Finally, we often refer to the \textit{Hamming weight} of a bit string.
This is defined as the number of 1-bits in the string.  Thus in
$\Cnbitsc{4}$, for example, the bit strings

$$0000, 0100, 0101,\,\text{and }1111$$
have Hamming weights $0, 1, 2,$ and $4$, respectively.

\section{The \wt{}}

The \wt{} $\wmat{}$ is a linear map from $\Cnbits$ to itself. Defined
using the standard basis, we have

\begin{define}[\wt{}]\label{define:walsh}
$$\wmat{} \vIss{i} = \sumIss{j} \negonepwr{i}{j} \vIss{j}$$
\end{define}

The matrix entries of $\wmat{}$ with respect to the standard basis are
powers of $-1$, where the exponent is the ``bit dot product''
$\bitDotProd{i}{j}$:

\begin{equation}\label{matrixElement}
\innerp{\wmat{}\vIss{i}}{\vIss{j}} = \negonepwr{i}{j}
\end{equation}
We show below that the \wt{} is the scalar multiple of a unitary linear
mapping, thus the image of an orthogonal basis under $\wmat{}$ is also
an orthogonal basis.

\begin{define}[\wb{}]
The \wb{} is the orthogonal basis obtained by applying the \wt{}
to the standard basis.
\end{define}

\begin{example}
We consider the case where $n = 2$.  To write the matrix for $\wmat{}$,
we use the basis order $00, 10, 01, 11$.  Then $\wmat{}$ becomes the
following self-adjoint matrix.

$$\begin{pmatrix*}[r]
1&1&1&1\\
1&-1&1&-1\\
1&1&-1&-1\\
1&-1&-1&1
\end{pmatrix*}$$
\end{example}

The \wb{} can then be read off from the rows (or columns, since
$\wmat{}$ is
self-adjoint and real) of the matrix.
Explicitly, it is the set

\begin{equation}\label{example n=4}
\begin{aligned}
\{ & \wBasisElt0,\\
     &\wBasisElt1,\\
     &\wBasisElt2,\\
     &\wBasisElt3\}
\end{aligned}
\end{equation}

\begin{remark}
It is worth noting that the \wt{} is the $n$-fold tensor product of the
Fourier transform (multiplied by a constant) on
$\mathbb{Z}/2\mathbb{Z}$.  However, we will not explore this avenue in
this article.  Readers with a background in abstract mathematics who are
familiar with tensor products and their various properties may find it
amusing to find shorter proofs than those presented here.  The
connection with the Fourier transform was also pointed out in
\cite{f622c051-cd98-3197-b98c-f4515ba3978d}.
\end{remark}

\section{Properties of the \wt{}}

We begin this section with the following theorem.

\begin{thm}\label{thm:fundamental walsh}
The \wt{} is
\begin{enumerate}[(a)]
\item Self-adjoint\label{self-adjoint}, i.e.\ $\wmat{}^* = \wmat{}$
\item Has square equal to $2^n$ times the identity, i.e.\label{orthogonal}

$$\wmat{}^2 = 2^n I$$
where $I$ is the identity mapping on $\Cnbits{}$.

\end{enumerate}
\end{thm}

\noindent
Note that it follows immediately from Theorem~\ref{thm:fundamental
walsh} that $U = \frac{1}{2^{n/2}}\wmat{}$ is unitary, i.e.\ $UU^{*} =
U^{*}U = I$.

\begin{proof}
Item \eqref{self-adjoint} is seen by noting that the inner product is
symmetric on vectors with real components and by considering the
following equations.

\begin{align*}
\innerp{\wmat{}\vIss{i}}{\vIss{j}} &= \negonepwr{i}{j}\quad\text{by
\eqref{matrixElement}}\\
&= \negonepwr{j}{i}\\
&= \innerp{\wmat{}\vIss{j}}{\vIss{i}}\quad\text{also by
\eqref{matrixElement}}\\
&= \innerp{\vIss{i}}{\wmat{}\vIss{j}}
\end{align*}
Since the standard basis spans $\Cnbits$, it follows
that $\wmat{}$ is self-adjoint.

Looking now at \eqref{orthogonal}, we see from \eqref{self-adjoint} that
$\wmat{}^2 = \wmat{} \wmat{}^*$.  Hence it suffices to show that the
rows of the matrix \eqref{matrixElement} form an orthogonal set of
vectors in $\Cnbits$ all with squared norm $2^n$.

The squared norm of each row is given by

$$\sumIss{j} \left|\negonepwr{i}{j}\right|^2 = \sumIss{j} 1,$$
which is obviously $2^n$, proving the claim about the norm of each row.

For the claim of orthogonality of the rows, let $\Iss{i}$ and
$\Iss{\otheri{i}}$ be two unequal bit strings corresponding to two
different rows.  The order of the bits is
just labeling, so we may assume that $i_1 \neq {\otheri{i}}_1$.  The dot
product of the two rows is

\begin{align*}
\sumIss{j} \negonepwr{i}{j} \negonepwr{\otheri{i}}{j} &=
    \sumIss{j} \negonepwraddl{i}{\otheri{i}}{j}
\end{align*}
The set of bit strings $\Iss{j}$ is the $n$-fold Cartesian product of
the set $\{0,1\}$.  Thus sums may be separated.  Rewriting, we have

\begin{multline}
\sumIss{j} \negonepwraddl{i}{\otheri{i}}{{j}} \\
          = \sum_{j_1 \in \{0,1\}} (-1)^{(i_1 + \otheri{i}_1)j_1}
    \sumIs{j}{2}\negonepwraddsl{i}{\otheri{i}}{j}{2}\hfil
\end{multline}
Since $i_1 + \otheri{i}_1 = 1$, the first sum on the right side of the
equation is zero, thus the entire expression is zero, demonstrating
orthogonality.
\end{proof}

\begin{remark}
For the remainder of the article, we adopt a notational convention.
When the \wt{} $\wmat{}$ is applied to a standard basis vector
$\vIss{i}$, we denote the image by applying a bold-faced `\textbf{w}'
superscript:

$$\wmat{} \vIss{i} = \wted{\vIss{i}}$$

For further clarity, we also the adopt the convention of using Greek
letters to denote bits whenever they are contained in an expression
with `\textbf{w}' superscript. Thus a \wb{}
vector might be written as
$\wted{\vIss{\alpha}}$ or $\wted{\vIss{\beta}}$, for
example.
\end{remark}

Finally, we follow the terminology of \cite{WEINREICH2013700} and call
the components of a vector written in terms of the \wb{}
\textit{\opnamecoeff}.

\section{An interpretation of the \wb{}}\label{interpretationSection}

We describe a statistical interpretation of the \wb{}.  This is a
trivial generalization from the biological use of epistasis with
bi-allelic systems to any control-treatment scenario involving several
treatments.

Indeed, we can think of each bit string $\Iss{i}$ as determining membership in
various ``control'' and ``treatment'' groups. Thus, for instance, the
bit string $10110$ would correspond to the group given treatments 1, 3,
and 4 but not treatments 2 and 5. A vector in $\Cnbits$ can be thought
of as the numerical results of all possible combinations of treatments.

For example, let $n = 2$, and let $v$ be a vector in $\Cnbitsc{2}$
corresponding to the results of an experiment involving two treatments.
The average effect of Treatment 1 may be expressed as follows.

$$\frac{ (v_{10} - v_{00} ) + ( v_{11} - v_{01} ) }{2}$$
Referring to \eqref{example n=4}, we see that this is the negative of
half the inner product of $v$ with one of the \wb{} vectors.

$$(v_{10} - v_{00} ) + ( v_{11} - v_{01} ) =
    -\innerp{v}{\wBasisElt1}$$
In general, for any bit string with Hamming weight $1$ with the 1-bit in
the $n$th position, its image under $\wmat{}$ corresponds to the average
effect of Treatment $n$.

How can we interpret the ``second order'' \wb{} vector in
\eqref{example n=4} (the last one)?  The \opnamecoef{} is given by

$$\innerp{v}{\wBasisElt3} =
    v_{00} - v_{01} - v_{10} + v_{11}$$
We can rewrite this as

$$
(v_{11} - v_{00}) - (v_{10} - v_{00}) - (v_{01} - v_{00})
$$
This is readily interpreted as ``The effect of both treatments on the
control group 00 after subtracting the effects of the individual
treatments on 00.'' We can see from this analysis how these
representations capture the notion of ``interaction'' among treatments,
where we wish to measure the effect of two treatments independently of
the individual effects.

A similar analysis can be made with third order or higher
\opnamecoeff{}.

\section{\Additive{} Models}

\Lsr{} models usually assume, prior to the inclusion of interaction
terms, that the output variable is a linear function of the input
variables plus noise.  The interpretation presented in the
\linmodAfterInterpret{} suggests the following definition.

\begin{define}\label{define linear}
We say that $v \in \Cnbits$ is \textit{additive}, or \textit{linear}, if
there exist constants $\effects$ such that for any
bit string $\Iss{i}$

$$v_{\Iss{i}} = \effected{i}$$
\end{define}

\noindent
Using the language of the \linmodAfterInterpret{}, this expresses that the
effect of a combination of treatments is the sum of the effects of the
individual treatments, i.e.\ no interactions.

The main result of this section is the following theorem.

\begin{thm}
A vector $v \in \Cnbits$ is \additive{} if and only if it is in the subspace
spanned by the \wb{} elements

\begin{equation}\label{linear basis elements}
\wted{\vx{00\ldots0}}, \wted{\vx{10\ldots0}},
\wted{\vx{01\ldots0}},\,\ldots \wted{\vx{00\ldots1}}
\end{equation}
i.e.\ the subspace spanned by the image under $\wmat{}$of the standard basis
elements with Hamming weight $\le 1$.
\end{thm}

\begin{proof}
\newcounter{linearproof}
It is clear from Definition \ref{define linear} that the set of
\additive{} vectors in $\Cnbits$ form a subspace of dimension $n + 1$.
Since the \wb{} elements in \eqref{linear basis elements} span a
subspace of equal dimension, it suffices to show that if $v$ is
additive, then for every bit string $\Iss{\alpha}$ with Hamming weight
$\ge 2$,

\begin{equation}\label{linearproof1}
\innerp{v}{\wted{\vIss{\alpha}}} = 0
\end{equation}

To see this, note that from Definition \ref{define linear} and \eqref{matrixElement} one
can write \eqref{linearproof1} as

\begin{equation}\label{linearproof2}
\sumIss{i} \negonepwr{i}{\alpha}(\effected{i})
\end{equation}
Defining $i_0 = 1$, and changing the order of the
summation, \eqref{linearproof2} can be written as

\begin{equation}\label{linearproof3}
\sum_{r=0}^n \sumIss{i} \negonepwr{i}{\alpha} c_{i_r} i_r
\end{equation}
Now consider a single term in the summation indexed by $r$.

\begin{equation}\label{rterm}
\sumIss{i} \negonepwr{i}{\alpha} c_{i_r} i_r
\end{equation}
From the assumption that $\Iss{\alpha}$ has Hamming weight $\ge 2$,
there must be an $s \ne r$ such that $\alpha_s = 1$.  Taking
$(-1)^{i_s\alpha_s}$ out as a separate summation, we rewrite
\eqref{rterm} as

\begin{multline*}
\sum_{i_s}(-1)^{i_s}
\sumIsk{i}{s} \negonepwrss{i}{\alpha}{1}{s-1} \negonepwrss{i}{\alpha}{s+1}{n} c_{i_r} i_r
\\
\end{multline*}
The leftmost summation is clearly equal to zero, hence the entire
expression is zero, as is the value of \eqref{linearproof2}.
\end{proof}

\section{The \zbb{}}

It is clear from Definition \ref{define:walsh} that the standard basis
expression of each \wb{} vector includes all the elements of the
standard basis among its terms, with coefficients either -1 or 1.
Researchers who consider using the Walsh basis to detect interactions
will find this disadvantageous given that having sufficient data for all
$2^n$ treatment combinations is unlikely when $n$ is large. This is
particularly true for geneticists, who are likely to be dealing with a
population where the set of detected genotypes may consist only of those
with a few out of a much larger set of possibilities.  In short, the
Hamming weight of bit strings representing treatment outcomes that one
\textit{sees} tends to be small.

An alternative is to use the \textit{\zbb{}}, a term we introduce here.
Like the \wb{}, it can be defined as the image of the standard basis
under a linear mapping $\zbmat{}$.  This basis has been used in
theoretical work, for example the authors of \cite{Dai_2013} do so
implicitly in their investigation of multivariate Bernoulli
distributions.

Before proceeding with the definition of $\zbmat{}$, we need the
following definition.

\begin{define}
Define the partial order $\po$ on bit strings of length $n$ as

\begin{equation*}\label{partialOrder}
\Iss{j} \po \Iss{i} \iff (\forall k\; j_k = 1 \implies i_k = 1)
\end{equation*}
Equivalently, if $\andop$ is the bit-wise ``and'' operation,
$$\Iss{j} \po \Iss{i} \iff \Iss{j} \andopSpaced \Iss{i} = \Iss{j}$$
\end{define}

Expressed differently yet again, $\Iss{j} \po \Iss{i}$ when $\Iss{j}$
has no bits set to 1 which aren't set to 1 in $\Iss{i}$, or rather the
support of $\Iss{j}$ lies in the support of $\Iss{i}$, where by
\textit{support} we mean the location of bits set to 1.

\begin{define}\label{0basedDef}
The \zbb{} is defined as the image of the standard basis under the
following linear transformation.

\begin{equation}\label{define:0based}
\zbmat \vIss{i} = \sum_{\Iss{j} \po \Iss{i}} \negonepwr{i}{j} \vIss{j}
\end{equation}
or, written as matrix elements with respect to the standard basis,
\begin{equation}\label{matrixElement0based}
\innerp{\zbmat \vIss{i}}{\vIss{j}} = \sum_{\Iss{j} \po \Iss{i}} \negonepwr{i}{j}
\end{equation}
\end{define}

Note the difference between $\zbmat{}$ and $\wmat$.  With $\zbmat{}$,
the sum is taken only over those bit strings $\Iss{j}$ whose support
lies in the support of $\Iss{i}$.

\begin{remark}
As with the operator $\wmat{}$, we adopt a notation using a superscript to
denote vectors of the form $\zbmat \vIss{i}$, using `\textbf{z}` rather
than `\textbf{w}'.  We use the Greek letter convention here as well.  Thus

$$\zbmat \vIss{\alpha} = \zbed{\vIss{\alpha}}$$
\end{remark}

\begin{example}
Let $n = 3$.  Here is a list of the \zbb{} vectors.

\begin{equation*}
\begin{aligned}
\zbed{\vx{000}} &= \vx{000}\\
\zbed{\vx{100}} &= \vx{000} - \vx{100}\\
\zbed{\vx{010}} &= \vx{000} - \vx{010}\\
\zbed{\vx{001}} &= \vx{000} - \vx{001}\\
\zbed{\vx{110}} &= \vx{000} - \vx{100} - \vx{010} + \vx{110}\\
\zbed{\vx{101}} &= \vx{000} - \vx{100} - \vx{001} + \vx{101}\\
\zbed{\vx{011}} &= \vx{000} - \vx{010} - \vx{001} + \vx{011}\\
\zbed{\vx{111}} &= \vx{000} - \vx{010} - \vx{010} - \vx{001}
+ \vx{110} + \vx{101} + \vx{011} - \vx{111}
\end{aligned}
\end{equation*}
Compare this to the \wb{} vectors with the same number of dimensions.

\begin{equation*}
\begin{aligned}
\wted{\vx{000}} &= \vx{000} + \vx{100} + \vx{010} + \vx{001} + \vx{110}
+ \vx{101} + \vx{011} + \vx{111}\\
\wted{\vx{100}} &= \vx{000} - \vx{100} + \vx{010} + \vx{001} - \vx{110}
- \vx{101} + \vx{011} - \vx{111}\\
\wted{\vx{010}} &= \vx{000} + \vx{100} - \vx{010} + \vx{001} - \vx{110}
+ \vx{101} - \vx{011} - \vx{111}\\
\wted{\vx{001}} &= \vx{000} + \vx{100} + \vx{010} - \vx{001} + \vx{110}
- \vx{101} - \vx{011} - \vx{111}\\
\wted{\vx{110}} &= \vx{000} - \vx{100} - \vx{010} + \vx{001} + \vx{110}
- \vx{101} - \vx{011} + \vx{111}\\
\wted{\vx{101}} &= \vx{000} - \vx{100} + \vx{010} - \vx{001} - \vx{110}
+ \vx{101} - \vx{011} + \vx{111}\\
\wted{\vx{011}} &= \vx{000} + \vx{100} - \vx{010} - \vx{001} - \vx{110}
- \vx{101} + \vx{011} + \vx{111}\\
\wted{\vx{111}} &= \vx{000} - \vx{100} - \vx{010} - \vx{001} + \vx{110}
+ \vx{101} + \vx{011} - \vx{111}\\
\end{aligned}
\end{equation*}
\end{example}

We show that the \zbb{} is in fact a basis by showing that $\zbmat{}$ is
invertible.  Indeed, it is its own inverse.

\begin{thm}\label{thm:fundamental 0based}
$\zbmat^2 = I$
\end{thm}

\begin{proof}
Expanding twice from Definition \ref{0basedDef}, we get for all
$\Iss{i}$

\begin{multline}\label{compactSum}
\zbmat^2\vIss{i} = \zbmat \sum_{\Iss{j} \po \Iss{i}} \negonepwr{i}{j}
\vIss{j}\\
= \sum_{\Iss{j} \po \Iss{i}} \negonepwr{i}{j}
\zbmat{\vIss{j}}\\
= \sum_{\Iss{j} \po \Iss{i}} \negonepwr{i}{j} \sum_{\Iss{k} \po
\Iss{j}} \negonepwr{j}{k} \vIss{k}\\
= \orderedSumn{k}{j}{i}
\negonepwradd{j}{i}{k} \vIss{k}
\end{multline}

In an effort to keep the summation notation comprehensible, we adopt the
convention that when an expression in a summation index is fixed
relative to the summation, then it is preceded by a raised $\dag$.

From \eqref{compactSum}, it suffices to
show that

\begin{equation}\label{inductionHyp}
\orderedSumnFixed{k}{j}{i}
\negonepwradd{j}{i}{k} = \delta_{\Iss{i},\,\Iss{k}}
\end{equation}
for any $\Iss{i}$, $\Iss{k}$.  We demonstrate this by induction on $n$.

The case $n = 0$ is trivial, indeed vacuous.  Assume
\eqref{inductionHyp} is valid:  We show that the case $n + 1$:

\begin{equation}\label{inductionConc}
\orderedSumxFixed{k}{j}{i}{n+1}
\negonepwraddss{j}{i}{k}{1}{n+1} = \delta_{\Is{i}{1}{n+1},\,\Is{k}{1}{n+1}}
\end{equation}
follows.

Consider a pair $\Is{i}{1}{n+1}$ and $\Is{k}{1}{n+1}$. If it is not the
case that $\Is{k}{1}{n+1} \po \Is{i}{1}{n+1}$, then
\eqref{inductionConc} is trivially true, since the summation has no
terms and is thus zero. So we assume it \textit{is} the case that
$\Is{k}{1}{n+1} \po \Is{i}{1}{n+1}$.  We consider all possible terms in
the summation which have the form

$$\potriple{k}{j}{i}{n+1}$$
If the first $n$ bits of each such pair are given, what are the
possibilities for the final bits, $i_{n+1}$ and $k_{n+1}$?  If they are
the same, then $j_{n+1}$ is squeezed between them, and so
$j_{n+1}(i_{n+1} + k_{n+1})$ is either 0 or 2, which does not affect the
power of -1, and which reduces both sides to \eqref{inductionConc} to
the induction hypothesis \eqref{inductionHyp}

The other possibility is that $k_{n+1} = 0$ and $i_{n+1} = 1$.  Then
both values of $j_{n+1}$ are possible, and will each occur once in the
summation (recall that the first $n$ bits of each bit string involved is
fixed).  When $j_{n+1} = 0$, we have $j_{n+1}(i_{n+1} + k_{n+1}) = 0$.
When $j_{n+1} = 1$, we have $j_{n+1}(i_{n+1} + k_{n+1}) = 1$. Thus the
pair of powers of -1 cancel each other, and the total sum
\eqref{inductionConc} is zero, as it should be since $\Is{i}{1}{n+1} \ne
\Is{k}{1}{n+1}$.
\end{proof}

\section{Relation with \lsr{}}\label{lastSection}

An alternative way of analyzing interactions is to use \lsr{}.  This
would treat the values for each bit string as data points, and consider
the binary input variables described by ``bit 1 set'', ``bit 2 set'',
etc.  In this section we provide evidence that regression and
\opnamecoeff{} are two ways of doing the same thing.

We offer no rigorous general proofs.  That would require going into the
technical details of regression algorithms.  Rather, we illustrate with
a toy example that the quantities of interest produced by a regression
analysis can be expressed in terms of \opnamecoeff{}, and vice versa.

For our regression analysis, we use two widely used Python language open
source packages: \pandas{}\footnote{Git version \pandaGit{}} to
construct data frames, and \statsmodels{}\footnote{Version
\statsmodelsVersion{}} to perform the \ls{} algorithm.  In parallel, we
show how the \opnamecoeff{} readily replicate the output of
\statsmodels{}' regression algorithm.

\begin{table}[ht]
\caption{Data example}\label{Data example}
\begin{tabular}{lr}
000: &100\\
100: &77\\
010: &115\\
001: &113\\
110: &93\\
101: &72\\
011: &129\\
111: &86\\
\end{tabular}\label{toyDataExample}
\end{table}

Our example data is shown in \tableref{} \ref{Data example}.  In the
context of regression, we view the list as the set of data points. We
have three input variables, \varstyle{bit1}, \varstyle{bit2}, and
\varstyle{bit3}, corresponding to whether \varstyle{bit}\textit{i} is
set. The output variable is the set of values in the right-hand column.
We call this variable \varstyle{val}.  In the listings beginning on page
\pageref{listings}, we record parts of a Python session. We use \pandas{}'
\texttt{DataFrame} constructor to gather the data in \tableref{}
\ref{toyDataExample} and form the data frame object in \codewin{}
\ref{DataFrame}

Next, in \codewin{} \ref{regression}, we fit a regression model to
predict \varstyle{val} as a linear function of the three \varstyle{bit} variables.
Here we use \statsmodels' ordinary \lsr{} routine and compute the
predicted value for each data point.

We show that the predicted values are readily computed from the
\opnamecoeff{} of this system. Using Definition \ref{define:walsh} along
with standard change of basis calculations from linear algebra, we can
rewrite

\begin{align*}
100 \vx{000} + 77 \vx{100} + 115 \vx{010} + 113 \vx{001}\\
+ 93\vx{110} + 72\vx{101}
+ 192\vx{011} + 86\vx{111}\hfil
\end{align*}
as
\begin{align}\label{walshedExample}
785\wted{\vx{000}} + 129\wted{\vx{100}} - 61\wted{\vx{010}} -
15\wted{\vx{001}}\\
- \wted{\vx{110}} - 39\wted{\vx{101}} -
\wted{\vx{011}} + 3\wted{\vx{111}}\nn
\end{align}
We then orthogonally project onto the subspace generated by $\vx{000}$,
$\vx{100}$, $\vx{010}$, and $\vx{001}$, with the view that these
elements correspond to the input variables used in the regression.  (In
nearly all applications of regression, here as well, the constant term
is an implicit input variable.)  The projected vector is

\begin{equation}\label{projectedVector}
785\wted{\vx{000}} + 129\wted{\vx{100}} - 61\wted{\vx{010}} -
15\wted{\vx{001}}
\end{equation}
Note that we are relying on the fact that the \wb{}
is orthogonal, so that setting components of it to zero constitutes an
orthogonal projection onto a subspace.

When rewritten using the standard basis, \eqref{projectedVector} becomes

\begin{align*} 104.75 \vx{000} + 72.5 \vx{100} + 120 \vx{010} + 108.5
\vx{001}\\ + 87.75\vx{110} + 76.25\vx{101} + 123.75\vx{011} +
91.5\vx{111},
\end{align*}
We note that the values which appear here are
precisely the predicted values gotten from the regression in \codewin{}
\ref{regression}.

We also consider the $R^2$ parameter for the regression shown at the
bottom of \codewin{} \ref{regression}.  Normally, the precision
displayed would be considered unnecessary and perhaps silly, but in this
case it is appropriate since we are checking whether deriving values
from the \opnamecoeff{} replicates precisely the result of the
algorithm.

We can recover $R^2$ from the \opnamecoeff{} by taking the quotient of
the sum of the squares of the one bit \opnamecoeff{} with the sum of the
squares of the non-constant \opnamecoeff{}.

\begin{equation*}
\frac{129^2 + 61^2 + 15^2}
    {129^2 + 61^2 + 15^2 + 110^2 + 39^2 + 1^2 + 3^2}
    = 0.9307382793073828,
\end{equation*}
which is exactly $R^2$ derived from the regression algorithm.

Now consider a regression model which includes two-way interactions,
shown in \codewin{} \ref{twoway}.  In a way similar to what we did with
the linear term-only model, we start with the \opnamecoeff{} in
\eqref{walshedExample} and project onto the subspace generated by the
constant, first order, and second order \wb{} elements:

\begin{equation*}
785\wted{\vx{000}} + 129\wted{\vx{100}} - 61\wted{\vx{010}} -
15\wted{\vx{001}} - \wted{\vx{110}} - 39\wted{\vx{101}} -
\wted{\vx{011}}
\end{equation*}
Reverting to the standard basis gives

\begin{align*}
99.625 \vx{000} + 77.375 \vx{100} + 115.375 \vx{010} + 113.375
\vx{001}\\
+ 92.625\vx{110} + 71.625\vx{101} + 128.625\vx{011} + 86.375\vx{111},\hfil
\end{align*}
which contains exactly the predicted values in \codewin{} \ref{twoway}

As before, we look at the $R^2$ generated by the regression algorithm.
It produces the value found at the bottom of \codewin{}
\ref{twoway}.  (See the previous comment about the precision.)

We can recover this though the \opnamecoeff{} by taking the quotient of
the sum of squares of all but the constant and third order
\opnamecoeff{} with the sum of the squares of all but the constant
\opnamecoef{}.

\begin{equation*}
\frac{129^2 + 61^2 + 15^2 + 1^2 + 39^2 + 1^2}{129^2 + 61^2 + 15^2 + 1^2 + 39^2
+ 1^2 + 3^2} = 0.9995931099959311,
\end{equation*}
which, as with the previous example, is exactly the $R^2$ generated by
algorithm.

\Bigskip

The guts of any \lsr{} calculation explicitly or implicitly involves the
orthogonalization of a matrix whose columns correspond to the input
variables of the model.  What the example in this section illustrates is
the connection between that process and the \wt{}.

There has been discussion in the \subject{} literature concerning the
usefulness of \opnamecoeff{} or interaction coordinates in general for
prediction and modeling \cite{sailer2018uninterpretable}. The point of
this section is not to take a side in this debate, but rather to suggest
that the discussion over the suitability of interaction coordinates may
be folded into the general discussion of regression models, with its
myriad concerns including over-fitting, model misspecification,
cross-validation, and so forth.

Furthermore, in a practical setting, unless there is a need for a
particular combination of interaction coordinates, the example above
suggests that an investigator dispense with the intricacies of the \wt{}
and its associated concepts and simply do a standard regression
analysis, including interactions when of if needed.

\section*{Conclusion}

The goal of this paper was to provide a concise mathematical development
of two forms of interaction coordinates.   In the \lastyn y section it
was suggested that the \wt{} may be viewed as part of the
orthogonalization procedure carried out as part of a linear regression
analysis with binary input variables.  In this light, debates on the
merits of considering interactions in modeling or prediction may be
merged with those arising in the general practice and principles of
regression analysis, leaving the matter to the statisticians with their
ample experience.

However, none of the above comments should be seen as implying that the
\wt{} is only of interest to mathematicians or curious theoreticians.
For an analogy, consider the Fourier transform.  To a signal engineer or
theoretical physicist, the Fourier transform is not merely a tool for
prediction and model building, but also a way of providing a useful
alternative representation of a signal or wave function. The intuition
which a signal engineer brings to a problem would be quite hobbled
without it.  In similar way, we suggest that the \wt{} can serve a similar
role to an investigation of the effects of binary treatment variables
and how their effect interact.  It provides a \textit{description} of
the results of a set of control-treatments groups, focusing not on the
numerical results for each group, but rather the interactions among the
various treatments.

\ack
Thanks to Kristina Crona for useful suggestions during the writing of
this article.

\pagebreak\label{listings}
\begin{lstlisting}[caption={Data frame}, label=DataFrame]
>>> df
     bit1  bit2  bit3  val
000     0     0     0  100
100     1     0     0   77
010     0     1     0  115
001     0     0     1  113
110     1     1     0   93
101     1     0     1   72
011     0     1     1  129
111     1     1     1   86
\end{lstlisting}
\vfil
\begin{lstlisting}[caption={Fitted linear model}, label={regression}]
>>> from statsmodels.api import OLS as ols
>>> linmod = ols.from_formula(
...         'val ~ bit1 + bit2 + bit3',
...         df)
>>> fitted = linmod.fit()
>>> fitted\
...         .get_prediction()\
...         .summary_frame()\
...         .loc[:, 'mean']
000    104.75
100     72.50
010    120.00
001    108.50
110     87.75
101     76.25
011    123.75
111     91.50
Name: mean, dtype: float64
>>> fitted.rsquared
0.9307382793073828
\end{lstlisting}
\vfil
\pagebreak
\begin{lstlisting}[caption={Fitted linear model: 2nd order}, label={twoway}]
>>> linmod2 = ols.from_formula(
...         'val ~ bit1 + bit2 + bit3 + '
...         'bit1:bit2 + bit1:bit3 + bit2:bit3',
...         df)
>>> fitted2 = linmod2.fit()
>>> fitted2\
...         .get_prediction()\
...         .summary_frame().loc[:, 'mean']
000     99.625
100     77.375
010    115.375
001    113.375
110     92.625
101     71.625
011    128.625
111     86.375
Name: mean, dtype: float64
>>> fitted2.rsquared
0.9995931099959311
\end{lstlisting}

\bibliographystyle{alpha}

\bibliography{walsh_transform}

\end{document}